\documentclass{article}
\usepackage{natbib}
\usepackage{jon-diagrams}
\usepackage{amsmath}
\usepackage{hyperref}
\hypersetup{colorlinks, citecolor=DarkSlateGray}

\usepackage{amsthm}
\usepackage[nameinlink,capitalize,noabbrev]{cleveref}

\usepackage[svgnames]{xcolor}

\usepackage{amssymb}

\usepackage{microtype}
\usepackage{titling}
\usepackage{titlesec}
\usepackage{abstract}

\usepackage[tt=false]{libertine}
\usepackage[libertine]{newtxmath}
\newcommand\mathscr[1]{\mathcal{#1}}

\usepackage{amssymb}

\usepackage{mathtools} 
\usepackage{array}

\usepackage{tikz}
\usepackage[colorinlistoftodos,prependcaption,color=FireBrick!10,bordercolor=Black!50]{todonotes}
\tikzset{notestyleraw/.append style={rounded corners = 1pt}}
\makeatletter
\renewcommand{\@todonotes@textsize}{\footnotesize\sffamily}
\makeatother

\usepackage{mleftright}
\usepackage{xparse}
\usepackage{stmaryrd}
\usepackage{mathpartir}

\mprset{sep=10pt}

\usepackage{amsthm}

\makeatletter
\def\rightharpoonupfill@{\arrowfill@\relbar\relbar\rightharpoonup}
\newcommand{\overrightharpoonup}{%
\mathpalette{\overarrow@\rightharpoonupfill@}}

\NewDocumentCommand{\etc}{m}{\overrightharpoonup{#1}}

\theoremstyle{plain}
\newtheorem{theorem}{Theorem}[section]
\newtheorem{lemma}[theorem]{Lemma}

\theoremstyle{definition}

\newtheorem{notation}[theorem]{Notation}

\theoremstyle{assumption}
\newtheorem{assumption}[theorem]{Assumption}
\newtheorem{convention}[theorem]{Convention}
\newtheorem{remark}[theorem]{Remark}

\usetikzlibrary{positioning}
\usetikzlibrary{backgrounds}
\usetikzlibrary{calc}
\usetikzlibrary{trees}
\usetikzlibrary{fit}
\usetikzlibrary{shadows}
\usetikzlibrary{shadows.blur}
\usetikzlibrary{shapes.symbols}

\NewDocumentCommand\Eta{}{\mathrm{H}}

\NewDocumentEnvironment{grammar}{}{%
  \begin{center}%
    \begin{tabular}{>{\itshape(}l<{)} >{$}l<{$} @{\quad $\Coloneqq$\quad } >{$}l<{$}}%
}{%
    \end{tabular}%
  \end{center}%
}

\NewDocumentCommand\NewPairedDelimiter{mmmm}{%
  \NewDocumentCommand#2{mmm}{%
    \IfNoValueTF{##2}
      {\IfBooleanTF{##1}
        {\mleft#3##3\mright#4}
        {#3##3#4}}
      {\mathopen{##2#3}##3\mathclose{##2#4}}%
  }
  \NewDocumentCommand#1{som}{#2{##1}{##2}{##3}}
}
\NewPairedDelimiter{\braces}{\rawbraces}{\{}{\}}
\NewPairedDelimiter{\parens}{\rawparens}{(}{)}
\NewPairedDelimiter{\brackets}{\rawbrackets}{[}{]}
\NewPairedDelimiter{\verts}{\rawverts}{\lvert}{\rvert}
\NewPairedDelimiter{\Verts}{\rawVerts}{\lVert}{\rVert}
\NewPairedDelimiter{\bbrackets}{\rawbbrackets}{\llbracket}{\rrbracket}

\NewDocumentCommand\Sem{m}{\bbrackets*{#1}}

\NewDocumentCommand\OB{}{\Sym{ob}}
\NewDocumentCommand\HOM{}{\Sym{hom}}
\NewDocumentCommand\HOMID{}{\Sym{id}}
\NewDocumentCommand\HOMCMP{}{\Sym{cmp}}
\NewDocumentCommand\Hom{m m}{#1\Rightarrow{}#2}
\NewDocumentCommand\HomId{}{\Sym{id}}
\NewDocumentCommand\HomCmp{m m}{#1 \circ #2}
\NewDocumentCommand\SigCAT{}{\mathbb{T}_{\Sym{cat}}}
\NewDocumentCommand\SigCWF{}{\mathbb{T}_{\Sym{cwf}}}
\NewDocumentCommand\SigMLTT{}{\mathbb{T}_{\Sym{cwf}}^{\PI}}
\NewDocumentCommand\TY{}{\Sym{Ty}}
\NewDocumentCommand\EL{}{\Sym{El}}
\NewDocumentCommand\EMP{}{\Sym{emp}}
\NewDocumentCommand\EXT{}{\Sym{ext}}
\NewDocumentCommand\TYACT{}{\Sym{Ty/act}}
\NewDocumentCommand\ELACT{}{\Sym{El/act}}
\NewDocumentCommand\PROJ{}{\Sym{p}}
\NewDocumentCommand\VAR{}{\Sym{q}}

\NewDocumentCommand\Ty{m}{\Sym{Ty}\parens{#1}}
\NewDocumentCommand\El{m m}{\Sym{El}\parens{#1\vdash #2}}

\NewDocumentCommand\Act{m m}{#2 \brackets{#1}}
\NewDocumentCommand\BANG{}{\boldsymbol{!}}
\NewDocumentCommand\Emp{}{\cdot}
\NewDocumentCommand\SNOC{}{\Sym{snoc}}
\NewDocumentCommand\Snoc{m m}{\langle{#1, #2}\rangle}
\NewDocumentCommand\LVL{}{\Sym{lvl}}
\NewDocumentCommand\LZ{}{\Sym{z}}
\NewDocumentCommand\LS{}{\Sym{s}}
\NewDocumentCommand\ITy{m m}{\TY_{#1}\parens{#2}}
\NewDocumentCommand\LT{}{\Sym{lt}}
\NewDocumentCommand\LTZ{}{\Sym{lt/z}}
\NewDocumentCommand\LTS{}{\Sym{lt/s}}
\NewDocumentCommand\LTC{}{\Sym{lt/lift}}
\NewDocumentCommand\LTCMP{}{\Sym{lt/cmp}}
\NewDocumentCommand\LIFT{}{\Sym{lift}}
\NewDocumentCommand\Lift{O{} O{} m}{{\Uparrow_{#1}^{#2}}#3}
\NewDocumentCommand\PI{}{\boldsymbol{\Pi}}
\NewDocumentCommand\LAM{}{\boldsymbol{\lambda}}
\NewDocumentCommand\APP{}{\Sym{app}}
\NewDocumentCommand\TyPi{m m}{\PI\parens{#1,#2}}
\NewDocumentCommand\App{m m}{\APP\parens{#1,#2}}
\NewDocumentCommand\Lam{m}{\LAM\parens{#1}}
\NewDocumentCommand\UNIV{}{\Sym{U}}
\NewDocumentCommand\Univ{m}{\UNIV_{#1}}

\NewDocumentCommand\Subject{O{\BooleanTrue} m}{
  \IfBooleanTF{#1}{{\color{DarkGreen} #2}}{#2}
}

\NewDocumentCommand\SORT{}{\mathit{sort}}
\NewDocumentCommand\GatJdg{m}{\mathrel{\vdash_{#1}}}

\NewDocumentCommand\IsThy{s m}{\Subject[#1]{#2}\ \mathit{thy}}
\NewDocumentCommand\IsTele{s O{\mathbb{T}} m}{\Subject[#1]{#3}\ \mathit{tele}_{#2}}

\NewDocumentCommand\IsSort{s O{\mathbb{T}} m m}{#3\GatJdg{#2} \Subject[#1]{#4}\ \SORT}
\NewDocumentCommand\EqSort{s O{\mathbb{T}} m m m}{#3\GatJdg{#2} #4 = #5\ \SORT}
\NewDocumentCommand\IsTerm{s O{\mathbb{T}} m m m}{#3\GatJdg{#2} \Subject[#1]{#4} : #5}
\NewDocumentCommand\IsSubst{s O{\mathbb{T}} m m m}{#3\GatJdg{#2} \Subject[#1]{#4} : #5}
\NewDocumentCommand\EqSubst{s O{\mathbb{T}} m m m m}{#3\GatJdg{#2} #4 = #5 : #6}
\NewDocumentCommand\EqTerm{s O{\mathbb{T}} m m m m}{#3\GatJdg{#2} #4 = #5 : #6}
\NewDocumentCommand\Cut{m m}{#1\langle{#2}\rangle}

\NewDocumentCommand\RuleSortDecl{O{} m m g}{
  \inferrule[#1]{
    #3
  }{
    \IfValueT{#4}{#4\ }\mathit{sort}
  }\ \tikz{\node[sort]{$#2$}}
}

\NewDocumentCommand\RuleOpDecl{O{} m m m g}{
  \inferrule[#1]{
    #3
  }{
    \IfValueT{#5}{#5\mathrel{:}}#4
  }
  \ \tikz{\node[opr]{$#2$}}
}

\NewDocumentCommand\RuleSortAxiom{O{} m m m}{
  \inferrule[#1]{
    #2
  }{
    #3 = #4\ \mathit{sort}
  }\ \tikz{\node[axiom]{$=$}}
}

\NewDocumentCommand\RuleTermAxiom{O{} m m m m}{
  \inferrule[#1]{
    #2
  }{
    #3 = #4 : #5
  }\ \tikz{\node[axiom]{$=$}}
}

\NewDocumentCommand\SortDecl{m}{\brackets*{#1\vdash\SORT}}
\NewDocumentCommand\OpDecl{m m}{\brackets*{#1\vdash #2}}
\NewDocumentCommand\SortAxiom{m m m}{\brackets*{#1 \vdash #2 = #3\ \SORT}}
\NewDocumentCommand\TermAxiom{m m m m}{\brackets*{#1 \vdash #2 = #3 : #4}}

\NewDocumentCommand\MSubst{m m}{#1^* #2}
\NewDocumentCommand\Dom{m}{\verts*{#1}}
\NewDocumentCommand\SigEmp{}{\varnothing}

\NewDocumentCommand\Id{}{\mathbf{id}}
\NewDocumentCommand\Sym{m}{\mathbf{#1}}
\NewDocumentCommand\GAT{}{\mathbf{GAT}}
\NewDocumentCommand\Alg{m}{#1\mbox{-}\mathbf{Alg}}

\NewDocumentCommand\OBS{}{\Sym{color}}
\NewDocumentCommand\RED{}{\Sym{red}}
\NewDocumentCommand\GREEN{}{\Sym{green}}
\NewDocumentCommand\CC{}{\mathcal{C}}
\NewDocumentCommand\Gl{m}{\overline{#1}}
\NewDocumentCommand\LFam{m}{#1^\bullet}
\NewDocumentCommand\Nat{}{\mathbb{N}}
\NewDocumentCommand\SemSort{O{}m g}{\llbracket{#2}\rrbracket_{#1}\IfValueT{#3}{\mkern-2mu\langle{#3}\rangle}}
\NewDocumentCommand\SemOp{O{}m g}{\llbracket{#2}\rrbracket_{#1}\IfValueT{#3}{\mkern-2mu\langle{#3}\rangle}}
\NewDocumentCommand\SemUniv{m}{\mathcal{V}_{#1}}

\NewDocumentCommand\Inl{m}{\mathbf{inl}\parens*{#1}}
\NewDocumentCommand\Inr{m}{\mathbf{inr}\parens*{#1}}

\NewDocumentCommand\HRule{}{%
  \nopagebreak[4]%
  \smallskip%
  \noindent\textcolor{Black!20}{\hrulefill}
}

\tikzset{
  background rectangle/.append style={color = Black!10, rounded corners = 2pt, fill = white},
  bend angle=45,
  presup/.style={rectangle, rounded corners=2pt, draw, signal to=south, draw=Black!30, fill = Black!2},
  jdg/.style={rectangle, rounded corners=2pt, draw, signal to=south, fill = Black!2},
  sort/.style={rectangle, rounded corners=2pt, draw, signal to=south, fill = RoyalBlue!10, text depth=0},
  opr/.style={rectangle, rounded corners=2pt, draw, signal to=south, fill = FireBrick!10, text depth=0},
  axiom/.style={rectangle, rounded corners=2pt, draw, signal to=south, fill = Gold!10, text depth=0},
}

\title{Algebraic Type Theory and Universe Hierarchies}

\author{
  Jonathan Sterling\thanks{\url{jmsterli@cs.cmu.edu}}
  \\
  {\small Carnegie Mellon University}
}

\begin{document}
\maketitle

\begin{abstract}

  It is commonly believed that algebraic notions of type theory support only
  universes \`a la Tarski, and that universes \`a la Russell must be removed by
  elaboration. We clarify the state of affairs, recalling the details of
  Cartmell's discipline of \emph{generalized algebraic theory}
  \citep{cartmell:1978}, showing how to formulate an algebraic version of
  Coquand's \emph{cumulative cwfs} with universes \`a la Russell.

  To demonstrate the power of algebraic techniques, we sketch a purely
  algebraic proof of canonicity for Martin-L\"of Type Theory with universes,
  dependent function types, and a base type with two constants.

\end{abstract}

\section{Generalized algebraic theories}

\citet{cartmell:1978} defines a notion of \emph{generalized algebraic theory},
generated by a collection of \emph{formation rules} for sort
symbols, \emph{introduction rules} for operation symbols, and axioms for both
sort equality and term equality.\footnote{\citet{taylor:1999} considers a
version of generalized algebraic theory which omits generating sort equations
(sort equality must still be considered as a consequence of sorts depending on
terms).} These are all inter-dependent, so care must be taken to stage the
construction properly.
In this note, we give a more streamlined presentation with a few differences from Cartmell's:
\begin{enumerate}

  \item At a superficial level, we use a modernized notation, inspired by
    logical frameworks.

  \item At a technical level, we stage the construction in such a way that,
    when considering the \emph{derived rules} available for a signature, we
    have already assumed that the signature is well-formed.

\end{enumerate}

It's worth noting that our formulation, while suitable for the development of
signatures which are finitary, is less general than the notion considered by
Cartmell; our version does not suffice to develop the universal algebra of
generalized algebraic theories (including the equivalence between the category
of such theories and the category of contextual categories), as noted by
\citet{taylor:1999}.

\subsection{Grammar and substitution}\label{sec:grammar}

In this section, we present an informal grammar to generate the raw syntax whose well-formedness we characterize in \cref{sec:gat-rules}.
\begin{grammar}
  judgments &
  \mathcal{J} &
  \!\!\!\begin{array}[t]{l}
    \IsThy{\mathbb{T}}
    \mid \IsTele{\Psi}
    \mid \IsSort{\Psi}{A}
    \mid \IsTerm{\Psi}{M}{A}
    \mid \IsSubst{\Phi}{\psi}{\Psi}
    \\
    \EqSort{\Psi}{A}{B}
    \mid \EqTerm{\Psi}{M}{N}{A}
    \mid \EqSubst{\Phi}{\psi_0}{\psi_1}{\Psi}
  \end{array}
  \\
  theories &
  \mathbb{T} &
    \SigEmp
    \mid \mathbb{T}, \vartheta:\mathcal{D}
    \mid \mathbb{T}, \mathcal{E}
  \\
  declarations &
  \mathcal{D} &
  \SortDecl{\Psi}
  \mid \OpDecl{\Psi}{A}
  \\
  axioms &
  \mathcal{E} &
  \SortAxiom{\Psi}{A}{B}
  \mid \TermAxiom{\Psi}{M}{N}{A}
  \\
  telescopes &
  \Psi &
  \cdot \mid \Psi, x : A
  \\
  sorts &
  A &
  \Cut{\vartheta}{\psi}
  \\
  terms &
  M &
  x \mid
  \Cut{\vartheta}{\psi}
  \\
  substitutions &
  \psi &
  \cdot \mid \psi, M/x
\end{grammar}

We define the action of raw substitutions on raw sorts and terms
$\MSubst{\psi}{(-)}$ and the composition of raw substitutions $\psi\circ\phi$
by recursion as follows:
\begin{gather*}
  \begin{aligned}
    \MSubst{\cdot}{x}
    &=
    x
    \\
    \MSubst{\parens*{\psi,M/x}}{x}
    &= x
    \\
    \MSubst{\parens*{\psi,M/y}}{x}
    &=
    \MSubst{\psi}{x}
    \\
    \MSubst{\psi}{
      \Cut{\vartheta}{\phi}
    }
    &=
    \Cut{\vartheta}{\phi\circ\psi}
  \end{aligned}
  \qquad
  \begin{aligned}
    \cdot\circ\psi &= \cdot
    \\
    \parens*{\phi,N/x}\circ\psi &=
    \parens*{\phi\circ\psi},\parens*{\MSubst{\psi}{N}}/x
  \end{aligned}
\end{gather*}

\subsection{Judgments and presuppositions}

We will define several forms of judgment simultaneously, to specify the
well-formedness and equality conditions for theories, sorts, terms and
substitutions. Following \citet{martin-lof:1996,schroeder-heister:1987}, we
explain a form of judgment by \emph{first} specifying its presuppositions (what
are the meaningful instances of the form of judgment?), and then giving rules
which specify when a meaningful instance of a form of judgment can be verified.

For instance, to specify a form of judgment like ``$M$ is of sort $A$'', we
first \emph{presuppose} that $A$ is already known to be a sort, but require of
$M$ only that it is generated from an appropriate production of the raw syntax
in \cref{sec:grammar}; then, rather than being false, the spurious instance
``$M$ is of sort $57$'' is actually not assigned a meaning at all. The
discipline of presuppositions enables us to omit many redundant premises from
rules.

\begin{convention}[Subjects and presuppositions]\label{convention:presupposition}

  In all cases that we will consider here, a form of judgment expresses that
  some piece of raw syntax (the \textbf{\textcolor{DarkGreen}{subject}}) is
  well-formed relative to some other objects which are already presupposed to
  be well-formed (the \textbf{parameters}); we write the subject in color to
  distinguish it visually from the parameters.
  We
  indicate this situation schematically for a form of judgment $\mathcal{J}$ in
  the following way:
  
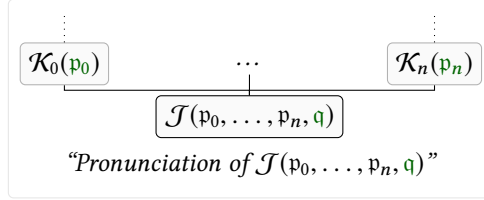
\[
    \tikz[framed,edge from parent fork up, level distance=2em, sibling distance = 7em]{
      \draw
        node[jdg] {$\mathcal{J}\parens*{\mathfrak{p}_0,\ldots,\mathfrak{p}_n,\Subject{\mathfrak{q}}}$} [grow'=up]
        child {
          node[presup] {$\mathcal{K}_0\parens*{\Subject{\mathfrak{p}_0}}$}
          child[dotted]
        }
        child{ node {\ldots} }
        child{
          node[presup] {$\mathcal{K}_n\parens*{\Subject{\mathfrak{p}_n}}$}
          child[dotted]
        }
        node[below=1em] {\textit{``Pronunciation of $\mathcal{J}\parens{\mathfrak{p}_0,\ldots,\mathfrak{p}_n,\Subject{\mathfrak{q}}}$''}}
    }
  \]

  The above schema should be read as asserting that the judgment
  $\mathcal{J}\parens*{\mathfrak{p}_0,\ldots,\mathfrak{p}_n,\Subject{\mathfrak{q}}}$
  presupposes $\mathcal{K}_0\parens*{\Subject{\mathfrak{p}_0}}$ through
  $\mathcal{K}_n\parens*{\Subject{\mathfrak{p}}_n}$, transitively presupposing whatever is
  presupposed by $\mathcal{K}_i\parens*{\Subject{\mathfrak{p}}_i}$, establishing the
  well-formedness of the raw syntax $\Subject{\mathfrak{q}}$.

\end{convention}

Once the meaningful instances of a form of judgment have been generated
schematically as in \cref{convention:presupposition}, its
\emph{correct} instances can be characterized inductively by rules of
inference.

\paragraph{Relation to traditional presentations}

The ``traditional'' presentation of rules of inference, in which \emph{all}
constituents are treated as subjects and inference rules are equipped with
extra premises to govern their well-formedness, can be obtained in a
completely mechanical way from the more streamlined systems presented here.

A form of judgment
$\mathcal{J}\parens*{\vec{\mathfrak{p}},\Subject{\mathfrak{q}}}$ should be
thought of as a family of sets of derivations indexed in
$\braces{\parens{\vec{\mathfrak{p}},\Subject{\mathfrak{q}}}\mid
\etc{\mathcal{K}\parens*{\Subject{\mathfrak{p}}}}}$. The act of forgetting the
well-formedness of the parameters induces a contravariant restriction of forms
of judgment in an adjoint situation:
\begin{diagram}
  \braces{\parens{\vec{\mathfrak{p}},\Subject{\mathfrak{q}}}\mid\etc{\mathcal{K}\parens*{\Subject{\mathfrak{p}}}}}
  &
  \mathbf{Jdg}\parens{\braces{\parens{\vec{\Subject{\mathfrak{p}}}, \Subject{\mathfrak{q}}}}}
  \\
  \dTo_i
  &
  \uTo^{\exists_i}\dashv\dTo~{i^*}\dashv\uTo_{\forall_i}
  \\
  \braces{\parens{\vec{\Subject{\mathfrak{p}}},\Subject{\mathfrak{q}}}}
  &
  \mathbf{Jdg}\parens{
    \braces{
      \parens{\vec{\mathfrak{p}},\Subject{\mathfrak{q}}}
      \mid\etc{\mathcal{K}\parens*{\Subject{\mathfrak{p}}}}
    }
  }
\end{diagram}

From the left adjoint $\exists_i$, we obtain exactly the system of judgments
and rules without presuppositions, in which the old presuppositions are added
as auxiliary premises to every rule in just the right place.

\subsection{Rules for generalized algebraic theories}\label{sec:gat-rules}

\begin{mathpar}
  \tikz[framed,edge from parent fork up, level distance=2em, sibling distance = 7em]{
    \draw
      node[jdg] {$\IsThy*{\mathbb{T}}$}
      node[below=1em] {\textit{``$\Subject{\mathbb{T}}$ is a theory''}}
  }
  \\
  \inferrule[empty]{
  }{
    \IsThy*{\SigEmp}
  }
  \and
  \inferrule[sort decl]{
    \IsThy*{\mathbb{T}}
    \\
    \IsTele*[\mathbb{T}]{\Psi}
    \\
    \parens*{\vartheta\not\in\mathbb{T}}
  }{
    \IsThy*{\mathbb{T}, \vartheta : \SortDecl{\Psi}}
  }
  \and
  \inferrule[operation decl]{
    \IsThy*{\mathbb{T}}
    \\
    \IsTele*[\mathbb{T}]{\Psi}
    \\
    \IsSort*[\mathbb{T}]{\Psi}{A}
    \\
    \parens*{\vartheta\not\in\mathbb{T}}
  }{
    \IsThy*{\mathbb{T}, \vartheta:\OpDecl{\Psi}{A}}
  }
  \and
  \inferrule[sort axiom]{
    \IsThy*{\mathbb{T}}
    \\
    \IsTele*[\mathbb{T}]{\Psi}
    \\
    \IsSort*[\mathbb{T}]{\Psi}{A_0}
    \\
    \IsSort*[\mathbb{T}]{\Psi}{A_1}
  }{
    \IsThy*{\mathbb{T},\SortAxiom{\Psi}{A_0}{A_1}}
  }
  \and
  \inferrule[term axiom]{
    \IsThy*{\mathbb{T}}
    \\
    \IsTele*[\mathbb{T}]{\Psi}
    \\
    \IsSort*[\mathbb{T}]{\Psi}{A}
    \\
    \IsTerm*[\mathbb{T}]{\Psi}{M_0}{A}
    \\
    \IsTerm*[\mathbb{T}]{\Psi}{M_1}{A}
  }{
    \IsThy*{\mathbb{T},\TermAxiom{\Psi}{M_0}{M_1}{A}}
  }
\end{mathpar}

\HRule

\begin{mathpar}
  \tikz[framed,edge from parent fork up, level distance=2em, sibling distance = 7em]{
    \draw
      node[jdg] {$\IsTele*{\Psi}$} [grow'=up]
      child {
        node[presup] { $\IsThy*{\mathbb{T}}$ }
      }
      node[below=1em] {\textit{``$\Subject{\Psi}$ is a telescope for theory $\mathbb{T}$''}}
  }
  \\
  \inferrule[empty]{
  }{
    \IsTele*{\cdot}
  }
  \and
  \inferrule[snoc]{
    \IsTele*{\Psi}
    \\
    \IsSort*{\Psi}{A}
    \\
    \parens*{x\not\in\Psi}
  }{
    \IsTele*{\Psi,x:A}
  }
\end{mathpar}

\HRule

\begin{mathpar}
  \tikz[framed,edge from parent fork up, level distance=2em, sibling distance = 7em]{
    \draw
      node[jdg] {$\IsSort*{\Psi}{A}$} [grow'=up]
      child {
        node[presup] { $\IsTele*{\Psi}$ }
        child[dotted] {}
      }
      node[below=1em] {\textit{``$\Subject{A}$ is a sort in context $\Psi$''}}
  }
  \\
  \inferrule[sort formation]{
    \mathbb{T}\ni\vartheta : \SortDecl{\Phi}
    \\
    \IsSubst*[\mathbb{T}]{\Psi}{\varphi}{\Phi}
  }{
    \IsSort*[\mathbb{T}]{\Psi}{\Cut{\vartheta}{\varphi}}
  }
\end{mathpar}

\HRule

\begin{mathpar}
  \tikz[framed,edge from parent fork up, level distance=2em, sibling distance = 7em]{
    \draw
      node[jdg] {$\EqSort*{\Psi}{A}{B}$} [grow'=up]
      child {
        node[presup] { $\IsSort*{\Psi}{A}$ }
        child[dotted] {}
      }
      child {
        node[presup] { $\IsSort*{\Psi}{B}$ }
        child[dotted] {}
      }
      node[below=1em] {\textit{``$A$ and $B$ are equal sorts''}}
  }
  \\
  \inferrule[reflexivity]{
  }{
    \EqSort*{\Psi}{A}{A}
  }
  \and
  \inferrule[symmetry]{
    \EqSort*{\Psi}{A_0}{A_1}
  }{
    \EqSort*{\Psi}{A_1}{A_0}
  }
  \and
  \inferrule[transitivity]{
    \EqSort*{\Psi}{A_0}{A_1}
    \\
    \EqSort*{\Psi}{A_1}{A_2}
  }{
    \EqSort*{\Psi}{A_0}{A_2}
  }
  \and
  \inferrule[substitution]{
    \EqSort*{\Phi}{A_0}{A_1}
    \\
    \EqSubst*{\Psi}{\phi_0}{\phi_1}{\Phi}
  }{
    \EqSort*{\Psi}{
      \MSubst{\phi_0}{A_0}
    }{
      \MSubst{\phi_1}{A_1}
    }
  }
  \and
  \inferrule[sort axiom]{
    \mathbb{T}\ni\SortAxiom{\Psi}{A_0}{A_1}
  }{
    \EqSort*[\mathbb{T}]{\Psi}{A_0}{A_1}
  }
\end{mathpar}

\HRule

\begin{mathpar}
  \tikz[framed,edge from parent fork up, level distance=2em, sibling distance = 7em]{
    \draw
      node[jdg] {$\IsTerm*{\Psi}{M}{A}$} [grow'=up]
      child {
        node[presup] { $\IsSort*{\Psi}{A}$ }
        child[dotted] {}
      }
      node[below=1em] {\textit{``$\Subject{M}$ is a term of sort $A$''}}
  }
  \\
  \inferrule[term formation]{
    \mathbb{T}\ni\vartheta:\OpDecl{\Phi}{A}
    \\
    \IsSubst*[\mathbb{T}]{\Psi}{\varphi}{\Phi}
  }{
    \IsTerm*[\mathbb{T}]{\Psi}{\Cut{\vartheta}{\varphi}}{\MSubst{\phi}{A}}
  }
  \and
  \inferrule[variable]{
  }{
    \IsTerm*{\Psi, x : A}{x}{A}
  }
  \and
  \inferrule[conversion]{
    \EqSort*{\Psi}{A_0}{A_1}
    \\
    \IsTerm*{\Psi}{M}{A_0}
  }{
    \IsTerm*{\Psi}{M}{A_1}
  }
\end{mathpar}

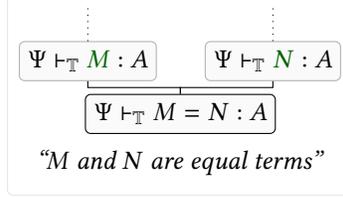
\HRule

\begin{mathpar}
  \tikz[framed,edge from parent fork up, level distance=2em, sibling distance = 7em]{
    \draw
      node[jdg] {$\EqTerm*{\Psi}{M}{N}{A}$} [grow'=up]
      child {
        node[presup] { $\IsTerm*{\Psi}{M}{A}$ }
        child[dotted] {}
      }
      child {
        node[presup] { $\IsTerm*{\Psi}{N}{A}$ }
        child[dotted] {}
      }
      node[below=1em] {\textit{``$M$ and $N$ are equal terms''}}
  }
  \\
  \inferrule[reflexivity]{
  }{
    \EqTerm*{\Psi}{M}{M}{A}
  }
  \and
  \inferrule[symmetry]{
    \EqTerm*{\Psi}{M_0}{M_1}{A}
  }{
    \EqTerm*{\Psi}{M_1}{M_0}{A}
  }
  \and
  \inferrule[transitivity]{
    \EqTerm*{\Psi}{M_0}{M_1}{A}
    \\
    \EqTerm*{\Psi}{M_1}{M_2}{A}
  }{
    \EqTerm*{\Psi}{M_0}{M_2}{A}
  }
  \and
  \inferrule[conversion]{
    \EqTerm*{\Psi}{M_0}{M_1}{A_0}
    \\
    \EqSort*{\Psi}{A_0}{A_1}
  }{
    \EqTerm*{\Psi}{M_0}{M_1}{A_1}
  }
  \and
  \inferrule[substitution]{
    \EqTerm*{\Phi}{M_0}{M_1}{A}
    \\
    \EqSubst*{\Psi}{\phi_0}{\phi_1}{\Phi}
  }{
    \EqTerm*{\Psi}{
      \MSubst{\phi_0}{M_0}
    }{
      \MSubst{\phi_1}{M_1}
    }{
      \MSubst{\phi_0}{A}
    }
  }
  \and
  \inferrule[term axiom]{
    \mathbb{T}\ni\TermAxiom{\Psi}{M_0}{M_1}{A}
  }{
    \EqTerm*[\mathbb{T}]{\Psi}{M_0}{M_1}{A}
  }
\end{mathpar}

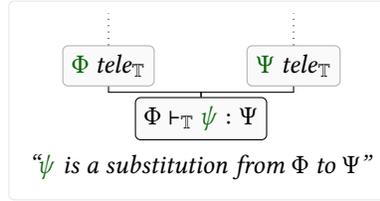
\HRule

\begin{mathpar}
  \tikz[framed,edge from parent fork up, level distance=2em, sibling distance = 7em]{
    \draw
      node[jdg] {$\IsSubst*{\Phi}{\psi}{\Psi}$} [grow'=up]
      child {
        node[presup] { $\IsTele*{\Phi}$ }
        child[dotted] {}
      }
      child {
        node[presup] { $\IsTele*{\Psi}$ }
        child[dotted] {}
      }
      node[below=1em] {\textit{``$\Subject{\psi}$ is a substitution from $\Phi$ to $\Psi$''}}
  }
  \\
  \inferrule[empty]{
  }{
    \IsSubst*{\Phi}{\cdot}{\cdot}
  }
  \and
  \inferrule[snoc]{
    \IsSubst*{\Phi}{\psi}{\Psi}
    \\
    \IsTerm*{\Phi}{M}{\MSubst{\psi}{A}}
  }{
    \IsSubst*{\Phi}{\parens*{\psi, M/x}}{\Psi,x:A}
  }
\end{mathpar}

\HRule

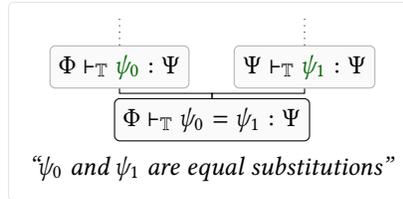
\begin{mathpar}
  \tikz[framed,edge from parent fork up, level distance=2em, sibling distance = 7em]{
    \draw
      node[jdg] {$\EqSubst*{\Phi}{\psi_0}{\psi_1}{\Psi}$} [grow'=up]
      child {
        node[presup] { $\IsSubst*{\Phi}{\psi_0}{\Psi}$ }
        child[dotted] {}
      }
      child {
        node[presup] { $\IsSubst*{\Psi}{\psi_1}{\Psi}$ }
        child[dotted] {}
      }
      node[below=1em] {\textit{``$\psi_0$ and $\psi_1$ are equal substitutions''}}
  }
  \\
  \inferrule[empty]{
  }{
    \EqSubst*{\Phi}{\cdot}{\cdot}{\cdot}
  }
  \and
  \inferrule[snoc]{
    \EqSubst*{\Phi}{\psi_0}{\psi_1}{\Psi}
    \\
    \EqTerm*{\Phi}{M_0}{M_1}{\Act{\psi_0}{A}}
  }{
    \EqSubst*{\Phi}{\parens*{\psi_0,M_0/x}}{\parens*{\psi_1,M_1/x}}{\Psi,x:A}
  }
\end{mathpar}

\HRule

\begin{notation}[Substitution]
  Abusing notation, we will often write $\psi,M$ instead of $\psi,M/x$ when the
  variable is clear from context; we will also routinely write $M$ instead of
  $\cdot,M$.
\end{notation}

\begin{lemma}[Substitution]
  When $\IsSubst*{\Phi}{\psi}{\Psi}$, we have the following admissible substitution principles:
  \begin{enumerate}
    \item If $\IsSort*{\Psi}{A}$, then $\IsSort*{\Phi}{\MSubst{\psi}{A}}$.
    \item If $\IsTerm*{\Psi}{M}{A}$, then $\IsTerm*{\Phi}{\MSubst{\psi}{M}}{\MSubst{\psi}{A}}$.
    \item If $\IsSubst*{\Psi}{\xi}{\Xi}$, then $\IsSubst*{\Phi}{\xi\circ\psi}{\Xi}$.
  \end{enumerate}
\end{lemma}

\begin{proof}
  By mutual induction on derivations.
\end{proof}

\subsection{Notation for theories}

Based on the rules and grammar that we have given, a generalized algebraic
theory is defined by a sequence of declarations of sort symbols and operation
symbols and their arities, with equational axioms interspersed, subject to the
sorting discipline. For instance, the generalized algebraic theory of monoids
is given as follows:
\begin{align*}
  &\Sym{ob}:\SortDecl{\cdot},
  \\
  &\Sym{id}:\OpDecl{\cdot}{\Cut{\Sym{ob}}{\cdot}},
  \\
  &\Sym{op}:\OpDecl{
    x : \Cut{\Sym{ob}}{\cdot},
    y : \Cut{\Sym{ob}}{\cdot}
  }{
    \Cut{\Sym{ob}}{\cdot}
  },
  \\
  &\TermAxiom{
    x : \Cut{\Sym{ob}}{\cdot}
  }{
    \Cut{\Sym{op}}{x,\Cut{\Sym{id}}{\cdot}}
  }{x}{
    \Cut{\Sym{ob}}{\cdot}
  },
  \\
  &\TermAxiom{
   x : \Cut{\Sym{ob}}{\cdot}
  }{
    \Cut{\Sym{op}}{\Cut{\Sym{id}}{\cdot},x}
  }{x}{
    \Cut{\Sym{ob}}{\cdot}
  },
  \\
  &\TermAxiom{
    x : \Cut{\Sym{ob}}{\cdot},y:\Cut{\Sym{ob}}{\cdot},z:\Cut{\Sym{ob}}{\cdot}
  }{
    \Cut{\Sym{op}}{\Cut{\Sym{op}}{x,y},z}
  }{
    \Cut{\Sym{op}}{x,\Cut{\Sym{op}}{y,z}}
  }{
    \Cut{\Sym{ob}}{\cdot}
  }
\end{align*}

However, for more complex theories, this linear notation will be a hindrance;
therefore, we will impose an inference-style notation which will be more
ergonomic. In our informal notation, a formation rule for a sort or operation
symbol will simultaneously extend the signature with the appropriate declaration, and
impose an informal \emph{notation} for its use. It is crucial to note that
these notations are just that: they are not part of the formal theory, but
instead part of the informal way that we render the real objects of the theory
(concrete trees) into linear text.

\paragraph{Sort declaration}

The sort formation rule
$
  \RuleSortDecl{\vartheta}{
    x_0 : A_0
    \quad
    \ldots
    \quad
    x_n : A_n
  }{\square_{x_0\ldots x_n}}
$
extends the signature by the declaration
$
  \vartheta:\SortDecl{
    x_0:A_0,\ldots,x_n:A_n
  }
$, and imposes the notational convention $\square_{x_0\ldots x_n}\triangleq \Cut{\vartheta}{x_0,\ldots,x_n}$.

These notational conventions are permitted to omit arguments which are obvious
from context, giving an informal counterpart to the notion of \emph{implicit
arguments} which appear in proof assistants like Agda~\citep{norell:2009}.

\paragraph{Operation declaration}
The operation formation rule
$
  \RuleOpDecl{\vartheta}{
    x_0 : A_0
    \quad
    \ldots
    \quad
    x_n : A_n
  }{B}{\square_{x_0\ldots x_n}}
$
extends the signature by the declaration
$\vartheta:\OpDecl{x_0:A_0,\ldots,x_n:A_n}{B}$, and imposes the notational
convention $\square_{x_0\ldots x_n} \triangleq\Cut{\vartheta}{x_0\ldots x_n}$.

\paragraph{Sort axiom}
The sort equation rule
$
  \RuleSortAxiom{
    x_0 : A_0
    \quad
    \ldots
    \quad
    x_n:A_n
  }{A}{B}
$
extends the signature by the axiom $\SortAxiom{x_0:A_0,\ldots,x_n:A_n}{A}{B}$.

\paragraph{Term axiom}
The term equation rule
$
  \RuleTermAxiom{
    x_0 : A_0
    \quad
    \ldots
    \quad
    x_n:A_n
  }{M}{N}{A}
$
extends the signature by the axiom $\TermAxiom{x_0:A_0,\ldots,x_n:A_n}{M}{N}{A}$.

\paragraph{Example} The theory of monoids can be written using our new notation
as follows:
\begin{mathpar}
  \RuleSortDecl{\Sym{ob}}{}{\Sym{ob}}
  \and
  \RuleOpDecl{\Sym{id}}{}{\Sym{ob}}{\Sym{id}}
  \and
  \RuleOpDecl{\Sym{cmp}}{
    x : \Sym{ob}
    \\
    y : \Sym{ob}
  }{
    \Sym{ob}
  }{x\bullet y}
  \and
  \RuleTermAxiom{
    x : \Sym{ob}
  }{
    x\bullet \Sym{id}
  }{x}{\Sym{ob}}
  \and
  \RuleTermAxiom{
    x : \Sym{ob}
  }{
    \Sym{id}\bullet x
  }{x}{\Sym{ob}}
  \and
  \RuleTermAxiom{
    x : \Sym{ob}
    \\
    y : \Sym{ob}
    \\
    z : \Sym{ob}
  }{
    \parens*{x\bullet y}\bullet z
  }{
    x \bullet\parens*{y \bullet z}
  }{
    \Sym{ob}
  }
\end{mathpar}

\subsection{Related work: logical frameworks}\label{sec:logical-frameworks}

Generalized algebraic theories comprise one point in the space of \emph{logical
frameworks}, which are syntactic disciplines for formulating deductive systems.
The purpose of a logical framework is to distinguish between the parts of a
deductive system which are \emph{particular} (for instance, the generators and
equations) and the parts which are \emph{universal} (for instance, the typing
or binding discipline). Logical frameworks vary primarily in which aspects of deductive
systems they treat as universal, negotiating the duality between
\emph{expressivity} and \emph{utility}.

\subsubsection{First-order algebraic theories}

One of the most basic logical frameworks is that of \emph{first-order algebra},
in which there are only atomic sorts, and contexts have the structure of a
strictly associative cartesian product. The science of functorial semantics was
developed first in the context of unisorted first-order algebraic theories by
\citet{lawvere:2004} in 1968.

Like generalized algebraic theories, first-order algebraic theories lack any
intrinsic binding structure. As such, they do not \emph{natively} explain the
universal syntactic phenomena which emanate from the lambda calculus and other
languages with binding.  In contrast, the dependent sorts of generalized
algebraic theories lend themselves to a workable formalization of De Bruijn
indices and explicit substitutions, which we have employed here.

\subsubsection{Second-order algebraic theories}

\citet{fiore-plotkin-turi:1999} initiated the scientific study of
\emph{second-order algebra} --- which had already appeared in embryonic form as
early as \citet{aczel:1978} --- a discipline which encompasses theories whose
operations exhibit binding structure of one level, with variables that range
over terms and metavariables which range over binders. The functorial semantics
of second-order algebraic theories was developed by \citet{fiore-mahmoud:2010}.
A variation on second-order algebra which omits both equations and
metavariables, but adds a novel notion of \emph{indexed operation}, was
employed by \citet{harper:2012:pfpl} to provide a syntactic discipline for
formulating the syntax and semantics of programming languages.

\subsubsection{Essentially algebraic theories}

Essentially algebraic theories are the closest (semantic) relative to
Cartmell's generalized algebraic theories. Generalized algebraic theories
provide dependently-sorted syntax for concepts which exhibit indexing;
essentially algebraic theories capture these concepts using \emph{fibration}
rather than \emph{parameterization}.

Every essentially algebraic theory can be presented as a generalized algebraic
theory by axiomatizing proof-irrelevant predicates; generalized algebraic
theories can likewise be transformed into essentially algebraic theories by
taking the ``total sorts'' of families and using predicates to specify indices.
When transforming an essentially algebraic theory into a generalized algebraic
theory, one must choose which relations to express using parameterization and
which to express using fibration. For this reason, Voevodsky observed that it
is not correct to consider the two disciplines
interchangeable~\citep{voevodsky:2013:tts}.

\subsubsection{Martin-L\"of's Logical Framework}

In \citet{martin-lof:1984}, the use of ``higher-level variables'' (variables
which range over binders) was introduced to the syntax of Intuitionistic Type
Theory, a syntactic discipline called a \emph{theory of expressions}
in~\citet{nordstrom-peterson-smith:1990}. This simply-typed higher-order
logical framework was employed to systematize the treatment of variable binding
in Intuitionistic Type Theory, which maintained at the time a separate
\emph{extrinsic} typing discipline which was defined on top of the simple
arities.

The theory of expressions was an immediate precursor to dependently typed
logical frameworks, principally Martin-L\"of's Logical Framework (LF) and the
Edinburgh Logical Framework (ELF). Object theories formulated in the LF use the
LF type structure to simultaneously express both their binding structure and
their typing discipline: constants are added to a signature together with an
LF-type, employing the ambient dependent function type to achieve both binding
(of higher level) and parameterization. Signatures in the LF can be extended
with equations, analogous to the state of affairs in generalized algebraic
theories.

\paragraph{Variable binding in LF}

Because object theories formulated in the LF inherit their binding discipline
from the metalanguage, there is likewise no need to formalize contexts. Whereas
in the generalized algebraic theory of categories with families (cwfs) which we
recapitulate in \cref{sec:cwf}, we formalize the notion of a context,
and then every operator takes as an argument its context $\Gamma$ in addition
to its other parameters, this part of the struture is implicit in the LF. A
consequence is, however, that the LF must be revised or extended in order to
support languages with exotic binding constructs and context effects, such as
modalities.

\subsubsection{Edinburgh Logical Framework}

Introduced by \citet{harper-honsell-plotkin:1993}, the Edinburgh Logical
Framework (ELF) shares its type structure with Martin-L\"of's Logical
Framework, but its notion of signature is different, and therefore its mode of
use also differs. Whereas the LF allowed a signature to be extended with
equations as is customary in algebra, the ELF was designed in order to enable a
strict bijection (called \emph{adequacy}) between ELF terms and object-theory
derivations, \emph{including} derivations of formal equality.

For this reason, while Martin-L\"of's Logical Framework is most suited to
providing a syntactic discipline for object theories which encompasses both
typing and formal equality, the Edinburgh Logical Framework is better adapted
to situations in which one wishes to prove a syntactic metatheorem about a
formal system by induction on its derivations.  The Edinburgh Logical
Framework, implemented in the Twelf proof
assistant~\citep{pfenning-schuermann:1999}, has been used to formalize the
syntactic metatheory of numerous logics and even programming
languages~\citep{lee-crary-harper:2007,harper-licata:2007}.

\subsubsection{Perspective}

The main axes of variation in logical frameworks are to be found in negotiating
the universality of \emph{type structure}, \emph{binding structure} and
\emph{formal equality}. In LF and ELF, the question of binding structure is
essentially subsumed by the type structure, which is treated as universal; in
LF, formal equality is treated as universal, whereas in ELF it is treated as
particular. Generalized algebraic theories represent a middle ground, in which
dependently-sorted first-order syntax can be formalized up to formal equality,
\emph{including} the theory of De Bruijn indices and explicit substitutions.

While such an object-level formalization of binding structure can prove
tedious, and necessarily results in a proliferation of axioms about how
substitutions propagate through constructors, it is the most natural setting in
which to develop an algebraic account of categories of models of type theory,
such as categories with families or categories with attributes. The initial
category with families (extended with further structure, such as dependent
function types), then, serves as a suitable notion of type theory; indeed, the
Logical Framework itself arises in this way.

The Logical Framework may be the most natural place to develop many object
languages, but we have found generalized algebraic theories to be a useful
intermediate point worth developing, if only to have a principled matrix in
which to construct the next thousand logical frameworks.

\subsection{Categories of signatures vs \emph{doctrines}}

In modern algebra, one considers two parallel perspectives on ``notion of
theory'':

\begin{enumerate}

  \item There is the 1-categorical perspective, in which a notion of theory
    is given by (something equivalent to) a 1-category of signatures and
    interpretations. This is the perspective that we have followed in
    \cref{sec:logical-frameworks}, and indeed, in the rest of this note.

  \item There is the 2-categorical perspective, in which a notion of theory is
    given by 2-category of theories (called a \emph{doctrine}). This usually
    arises from a \emph{universal characterization} of the basic structure
    (like finite products, finite limits, etc.), as opposed to a
    choice of basic structure.

\end{enumerate}

For instance, the 2-categorical perspective on algebraic theories arises from
the doctrine of categories with finite products and functors which preserve
finite products. On the other hand, one can define the 1-category of
first-order algebraic theories as either the category of algebraic signatures
and interpretations $\mathbf{Sign}_\times$, or the 1-category of categories
equipped with a choice of strictly-associative finite products $\mathbf{Law}$.
Both $\mathbf{Sign}_\times$ and $\mathbf{Law}$ are equivalent; while
$\mathbf{Law}$ looks as though it might extend to a 2-categorical notion
(because its objects are categories), it is fundamentally 1-categorical
because the only sensible class of 2-cells would contain only
identities~\citep{hyland-power:2007}.

\citet{cartmell:1978} treats generalized algebraic theories from a
purely 1-categorical perspective. Theories $\mathbb{T}$ are arranged into a
1-category $\GAT$, with morphisms $\mathbb{T}\to\mathbb{T}'$ given by
interpretations of the language of $\mathbb{T}$ into the language of
$\mathbb{T}'$; another 1-category of \emph{contextual categories} $\mathbf{Con}$
 is defined which plays exactly the role of
$\mathbf{Law}$ in relation to $\GAT$, which is to forget the difference between
derived and generating morphisms.

The 2-categorical perspective given by doctrines is essential for general
semantics, but our immediate aim is more restricted: we are using semantic
tools to prove theorems about syntax (such as canonicity, normalization,
coherence, decidability of equality, etc.); toward these ends, the
1-categorical perspective is the simplest and most immediately adaptable, and
we require only 1-categorical initiality results.

In contrast, to prove the equivalence of categories with families and locally
cartesian closed categories, one must develop the 2-categorical notion of cwf
rather than the 1-categorical notion that we develop here (i.e. the category of
models of the theory of cwfs); one does not obtain, for instance, a free
locally cartesian closed category from the initial model of the theory of
cwfs~\citep{castellan-clairambault-dybjer:2017}.

\subsection{Algebraic semantics and initiality}\label{sec:initiality}

Every generalized algebraic theory $\mathbb{T}$ gives rise to a category of models
$\Alg{\mathbb{T}}$; concretely, a model of $\mathbb{T}$ is given by an interpretation of
sorts and operations in families of sets. A sort ``$\IsSort*{\cdot}{A}$'' is
interpreted as a set $\Sem{A}$, whereas a sort ``$\IsSort*{x:A}{B}$'' is interpreted
as a $\Sem{A}$-indexed family of sets $\Sem{B}$, and so on.

Every interpretation $I:\mathbb{T}\to\mathbb{T}'$ in $\GAT$ induces a
restriction of algebras $I^* :\Alg{\mathbb{T}'}\to\Alg{\mathbb{T}}$
(precomposition with the interpretation), and this has a left adjoint $I_!$.
Considering the universal interpretation $!_{\mathbb{T}}:\SigEmp\to\mathbb{T}$,
we observe that the codomain $\Alg{\SigEmp}$ of its restriction functor is
actually the terminal category: there is only one model of the theory with no
sorts and no operations; from the left adjoint, we therefore obtain an initial
object in $\Alg{\mathbb{T}}$, i.e.\ an initial model.

The construction of the left adjoint involves adjoining new rules to the
\emph{term model} (Lindenbaum-Tarski model) of the given theory;
\citet{cartmell:1978} constructs this term model in painstaking detail, but
merely observes without proof that the left adjoint to the evident restriction
functor exists.

The existence of these left adjoints is proved in more detail in the context of
\emph{essentially algebraic theories} by \citet{palmgren-vickers:2007}. We are
not aware of a similarly detailed proof for generalized algebraic theories in
the literature.  If one is unsatisfied with this state of affairs, one can
observe that every generalized algebraic theory induces an essentially
algebraic theory with an equivalent category of models, and then transport the
initial object along this equivalence.

In recent work, \citet{kaposi-kovacs-altenkirch:2019} present a more modern
account of generalized algebraic theories in terms of \emph{finitary quotient
inductive types}, achieving a more crisp construction of initial algebras. As
with the presentation of generalized algebraic theories in \citet{taylor:1999}
(in contrast to \citeauthor{cartmell:1978}), they do not allow equations on
sorts. Lacking sort equations, we cannot reproduce the encodings in this paper
in an identical way, but we can recover the essence by encoding the theory of
GATs with sort equations as a GAT without sort equations.

\section{Warming up: the theory of categories}

We define the generalized algebraic theory of categories, $\IsThy{\SigCAT}$.
\begin{mathparpagebreakable}
  \RuleSortDecl{\OB}{}
  \and
  \RuleSortDecl{\HOM}{
    \Delta:\OB
    \\
    \Gamma:\OB
  }{
    \Hom{\Delta}{\Gamma}
  }
  \and
  \RuleOpDecl{\HOMID}{
    \Gamma:\OB
  }{
    \Hom{\Gamma}{\Gamma}
  }{
    \HomId
  }
  \and
  \RuleOpDecl{\HOMCMP}{
    \Eta,\Delta,\Gamma:\OB
    \\
    \gamma:\Hom{\Delta}{\Gamma}
    \\
    \delta:\Hom{\Eta}{\Delta}
  }{
    \Hom{\Eta}{\Gamma}
  }{
    \HomCmp{\gamma}{\delta}
  }
  \and
  \RuleTermAxiom{
    \Delta,\Gamma:\OB
    \\
    \gamma : \Hom{\Delta}{\Gamma}
  }{
    \HomCmp{\gamma}{\HomId}
  }{
    \gamma
  }{
    \Hom{\Delta}{\Gamma}
  }
  \and
  \RuleTermAxiom{
    \Delta,\Gamma:\OB
    \\
    \gamma : \Hom{\Delta}{\Gamma}
  }{
    \HomCmp{\HomId}{\gamma}
  }{
    \gamma
  }{
    \Hom{\Delta}{\Gamma}
  }
  \and
  \RuleTermAxiom{
    \Xi,\Eta,\Delta,\Gamma:\OB
    \\
    \eta : \Hom{\Xi}{\Eta}
    \\
    \delta : \Hom{\Eta}{\Delta}
    \\
    \gamma: \Hom{\Delta}{\Gamma}
  }{
    \HomCmp{
      \parens*{\HomCmp{\gamma}{\delta}}
    }{\eta}
  }{
    \HomCmp{\gamma}{
      \parens*{\HomCmp{\delta}{\eta}}
    }
  }{
    \Hom{\Xi}{\Gamma}
  }
\end{mathparpagebreakable}

The collection of models of $\SigCAT{}$ induces a \emph{1-categorical} notion
of ``category'', which we will exploit in our formulation of algebraic cwfs.

\section{Algebraic cwfs as a notion of type theory}\label{sec:cwf}

In semantics, categories with families (cwfs) are a familiar doctrine for type
theories~\citep{dybjer:1996,fiore:2012}, naturally organized into a 2-category;
but another perspective is given by the \emph{generalized algebraic theory of
cwfs}, whose 1-category of models and homomorphisms gives a more algebraic and
strict notion of type theory, in which all structure is chosen globally and
homomorphisms are arranged to preserve it on the nose.
We define the theory of cwfs, $\IsThy{\SigCWF}$, by extending
$\SigCAT$ with the following sorts, operations and axioms.

\paragraph{Types and elements}
\hfill
\nopagebreak[4]

\begin{mathparpagebreakable}
  \RuleSortDecl{\TY}{
    \Gamma:\OB
  }{\Ty{\Gamma}}
  \and
  \RuleOpDecl{\TYACT}{
    \Delta,\Gamma : \OB
    \\
    \gamma:\Hom{\Delta}{\Gamma}
    \\
    A : \Ty{\Gamma}
  }{
    \Ty{\Delta}
  }{
    \Act{\gamma}{A}
  }
  \and
  \RuleTermAxiom{
    \Gamma : \OB
    \\
    A : \Ty{\Gamma}
  }{
    \Act{\HomId}{A}
  }{A}{
    \Ty{\Gamma}
  }
  \and
  \RuleTermAxiom{
    \Eta,\Delta,\Gamma:\OB
    \\
    \delta:\Hom{\Eta}{\Delta}
    \\
    \gamma:\Hom{\Delta}{\Gamma}
    \\
    A : \Ty{\Gamma}
  }{
    \Act{\HomCmp{\gamma}{\delta}}{A}
  }{
    \Act{\delta}{\Act{\gamma}{A}}
  }{
    \Ty{\Eta}
  }
  \and
  \RuleSortDecl{\EL}{
    \Gamma:\OB
    \\
    A : \Ty{\Gamma}
  }{\El{\Gamma}{A}}
  \and
  \RuleOpDecl{\ELACT}{
    \Delta,\Gamma:\OB
    \\
    \gamma:\Hom{\Delta}{\Gamma}
    \\
    A : \Ty{\Gamma}
    \\
    M : \El{\Gamma}{A}
  }{
    \El{\Delta}{\Act{\gamma}{A}}
  }{\Act{\gamma}{M}}
  \and
  \RuleTermAxiom{
    \Gamma:\OB
    \\
    A : \Ty{\Gamma}
    \\
    M : \El{\Gamma}{A}
  }{
    \Act{\Id}{M}
  }{
    M
  }{
    \El{\Gamma}{M}
  }
  \and
  \RuleTermAxiom{
    \Eta,\Delta,\Gamma:\OB
    \\
    \delta:\Hom{\Eta}{\Delta}
    \\
    \gamma:\Hom{\Delta}{\Gamma}
    \\
    A : \Ty{\Gamma}
    \\
    a : \El{\Gamma}{A}
  }{
    \Act{\HomCmp{\gamma}{\delta}}{a}
  }{
    \Act{\delta}{\Act{\gamma}{a}}
  }{
    \El{\Eta}{\Act{\HomCmp{\gamma}{\delta}}{A}}
  }
\end{mathparpagebreakable}

\paragraph{Terminal context}

\begin{mathpar}
  \RuleOpDecl{\EMP}{}{\OB}{\cdot}
  \and
  \RuleOpDecl{\BANG}{
    \Gamma:\OB
  }{
    \Hom{\Gamma}{\Emp}
  }{\BANG}
  \and
  \RuleTermAxiom{
    \Delta,\Gamma:\OB
    \\
    \gamma:\Hom{\Delta}{\Gamma}
  }{
    \HomCmp{\BANG}{\gamma}
  }{
    \BANG
  }{
    \Hom{\Gamma}{\Emp}
  }
  \and
  \RuleTermAxiom{}{
    \HomId
  }{
    \BANG
  }{
    \Hom{\Emp}{\Emp}
  }
\end{mathpar}

\paragraph{Context comprehension}\hfill

\begin{mathparpagebreakable}
  \RuleOpDecl{\EXT}{
    \Gamma:\OB
    \\
    A : \Ty{\Gamma}
  }{
    \OB
  }{\Gamma.A}
  \and
  \RuleOpDecl{\SNOC}{
    \Delta,\Gamma:\OB
    \\
    A:\Ty{\Gamma}
    \\
    \gamma:\Hom{\Delta}{\Gamma}
    \\
    N : \El{\Delta}{\Act{\gamma}{A}}
  }{
    \Hom{\Delta}{\Gamma.A}
  }{\Snoc{\gamma}{N}}
  \and
  \RuleOpDecl{\PROJ}{
    \Gamma:\OB
    \\
    A:\Ty{\Gamma}
  }{
    \Hom{\Gamma.A}{\Gamma}
  }{\PROJ}
  \and
  \RuleOpDecl{\VAR}{
    \Gamma:\OB
    \\
    A:\Ty{\Gamma}
  }{
    \El{\Gamma.A}{\Act{\PROJ}{A}}
  }{\VAR}
  \and
  \RuleTermAxiom{
    \Delta,\Gamma:\OB
    \\
    \gamma:\Hom{\Delta}{\Gamma}
    \\
    A:\Ty{\Gamma}
    \\
    M:\El{\Delta}{\Act{\gamma}{A}}
  }{
    \HomCmp{\PROJ}{\Snoc{\gamma}{M}}
  }{
    \gamma
  }{
    \Hom{\Delta}{\Gamma}
  }
  \and
  \RuleTermAxiom{
    \Delta,\Gamma:\OB
    \\
    \gamma:\Hom{\Delta}{\Gamma}
    \\
    A:\Ty{\Gamma}
    \\
    M:\El{\Delta}{\Act{\gamma}{A}}
  }{
    \Act{\Snoc{\gamma}{M}}{\VAR}
  }{
    M
  }{
    \El{\Delta}{\Act{\gamma}{A}}
  }
  \and
  \RuleTermAxiom{
    \Eta,\Delta,\Gamma:\OB
    \\
    \gamma:\Hom{\Delta}{\Gamma}
    \\
    \delta:\Hom{\Eta}{\Delta}
    \\
    A:\Ty{\Gamma}
    \\
    M:\El{\Delta}{\Act{\gamma}{A}}
  }{
    \HomCmp{
      \Snoc{\gamma}{M}
    }{
      \delta
    }
  }{
    \Snoc{
      \HomCmp{\gamma}{\delta}
    }{
      \Act{\delta}{M}
    }
  }{
    \Hom{\Eta}{\Gamma.A}
  }
  \and
  \RuleTermAxiom{
    \Gamma:\OB
    \\
    A:\Ty{\Gamma}
  }{
    \Id
  }{
    \Snoc{\PROJ}{\VAR}
  }{
    \Hom{\Gamma.A}{\Gamma.A}
  }
\end{mathparpagebreakable}

The highly general results of \citeauthor{cartmell:1978} give rise to a
category of models of $\SigCWF$, which has an initial object, inducing a
functorial semantics in the sense of \citet{lawvere:2004}. This initial object is the
\emph{free Martin-L\"of Type Theory without any connectives}, and with a bit of
labor, it can be seen to be isomorphic to the Lindenbaum-Tarski model generated
by the raw syntax of pure MLTT without connectives.

By extending the theory $\SigCWF$ with further structure, such as dependent
function types, dependent pair types, universes, identity types, cubical
interval, etc., nearly every conceivable extension of MLTT can be obtained,
together with the appropriate category of models and homomorphisms. In this
note, part of our intention is to show how to obtain some of these extensions
which are commonly believed to fall outside the range of applicability of algebraic
techniques.

We emphasize that establishing the equivalence between these initial models and
Lindenbaum-Tarski models obtained by constraining and then quotienting the
``raw syntax'' is laborious in a technical sense, and may in the future be seen
to be superfluous: while some have fetishized the raw syntax of type theory to
such a degree that the specific textual/linear rendering of name binding (and
the attendant $\alpha$-convention) has been elevated from an expedient notation
to an actual object of study, we predict that the (abstractly presented)
initial models for algebraic type theory will ultimately be taken as definitive
in the study of type-theoretic syntax, as advocated
in~\citet{castellan-clairambault-dybjer:2017}.\footnote{\citeauthor{castellan-clairambault-dybjer:2017}
argue that it is circular to obtain the initial cwf from the generalized
algebraic theory of cwfs, because ``the notion of a generalised algebraic
theory is itself based on dependent type theory''. On the contrary, the notion
of generalized algebraic theory is developed primitively by
\citet{cartmell:1978,cartmell:1986} in the ambient set theory without making
use of any pre-existing type-theoretic machinery. Therefore, we maintain that
no circularity ensues from the algebraic generation of the initial cwfs;
\emph{a posteriori}, the latent cwf structure involved in defining the notion
of generalized algebraic theories can be observed, underscoring the unity
between the metatheory of type theory and type theory itself.}

Traditional presentations of type theory using raw syntax and unstructured
masses of inference rules took hold in the years before a workable account of
dependently typed universal algebra had been obtained; rather than making a
virtue out of ancient necessity, we hold that the most suitable notion of
syntax arises abstractly in a purely algebraic way (much like how previous
generations of type theorists had already eschewed the antique identification
of syntax with punctuated sequences of symbols~\citep{aczel:1978}).

From the perspective of a \emph{user} of type theory, we stress, there is no
serious gap presented by the abstract syntax induced by the initial cwf; in
fact, concrete computerized implementations of type theory tend to be much
closer to the abstract syntax of the initial cwf than to the ``raw'' syntax
which some have insisted is a primary object of study.

\section{Algebraic type hierarchies}

We begin by showing how to extend the theory of cwfs to include predicative
hierarchies of type systems, giving an algebraic treatment to Coquand's
\emph{cumulative cwfs}~\citep{coquand:2018}; then we will show how to add a
hierarchy of universes \`a la Russell in a modular way. We will replace the
$\TY$ operator with something parameterized in a universe level $\alpha$, giving
the sort of types of level $\alpha$.

\paragraph{Theory of type levels}
\hfill

\begin{mathparpagebreakable}
  \RuleSortDecl{\LVL}{}{\LVL}
  \and
  \RuleOpDecl{\LZ}{}{\LVL}{0}
  \and
  \RuleOpDecl{\LS}{
    \alpha:\LVL
  }{\LVL}{\alpha+1}
  \and
  \RuleSortDecl{\LT}{
    \alpha,\beta:\LVL
  }{\alpha<\beta}
  \and
  \RuleOpDecl{\LTZ}{
    \alpha:\LVL
  }{
    0 < \alpha + 1
  }
  \and
  \RuleOpDecl{\LTS}{
    \alpha:\LVL
  }{
    \alpha<\alpha+1
  }
  \and
  \RuleOpDecl{\LTC}{
    \alpha,\beta:\LVL
    \\
    p:\alpha<\beta
  }{
    \alpha+1<\beta+1
  }
  \and
  \RuleOpDecl{\LTCMP}{
    \alpha,\alpha',\beta:\LVL
    \\
    p:\alpha<\alpha'
    \\
    q:\alpha'<\beta
  }{
    \alpha<\beta
  }
  \and
  \RuleTermAxiom{
    \alpha,\beta:\LVL
    \\
    p,q:\alpha<\beta
  }{p}{q}{
    \alpha<\beta
  }
\end{mathparpagebreakable}

\paragraph{Types and elements}\hfill
\begin{mathparpagebreakable}
  \RuleSortDecl{\TY}{
    \alpha:\LVL
    \\
    \Gamma:\OB
  }{\ITy{\alpha}{\Gamma}}
  \and
  \RuleSortDecl{\EL}{
    \alpha:\LVL
    \\
    \Gamma:\OB
    \\
    A : \ITy{\alpha}{\Gamma}
  }{\El{\Gamma}{A}}
  \and
  \ldots
\end{mathparpagebreakable}

We will not recapitulate the remainder of the theory (e.g.\ context
comprehension), noting that it proceeds by adding $\alpha:\LVL$ to most
telescopes. Instead, we focus on what must be added to achieve the
\emph{algebraic} version of cumulativity for types, and then (algebraically)
cumulative universes \`a la Russell.

\begin{remark}[Algebraic cumulativity]

  In this note, we consider an algebraic form of \emph{cumulativity} which does
  not require any kind of subtyping. Instead, we have explicit shifts between
  universes \emph{which are ensured by algebraic laws to commute with all the
  connectives of type theory}: to put it crudely, we require
  $\Lift{A\times B} = \Lift{A}\times\Lift{B}$.

\end{remark}

\subsection{Algebraic cumulativity and lifting}

In order to achieve cumulativity, we add an operator which lifts
a type of level $\alpha$ to level $\beta>\alpha$:
\begin{mathparpagebreakable}
  \RuleOpDecl[lift formation]{\LIFT}{
    \alpha,\beta:\LVL
    \\
    p:\alpha<\beta
    \\
    \Gamma:\OB
    \\
    A:\ITy{\alpha}{\Gamma}
  }{
    \ITy{\beta}{\Gamma}
  }{\Lift[\alpha][\beta]{A}}
  \and
  \RuleTermAxiom[lift substitution]{
    \alpha,\beta:\LVL
    \\
    p : \alpha<\beta
    \\
    \Delta,\Gamma:\OB
    \\
    \gamma:\Hom{\Delta}{\Gamma}
    \\
    A:\ITy{\alpha}{\Gamma}
  }{
    \Act{\gamma}{\parens*{\Lift[\alpha][\beta]{A}}}
  }{
    \Lift[\alpha][\beta]{\parens*{\Act{\gamma}{A}}}
  }{
    \ITy{\alpha}{\Delta}
  }
  \and
  \RuleTermAxiom[lift composition]{
    \alpha,\alpha',\beta:\LVL
    \\
    p:\alpha<\alpha'
    \\
    q:\alpha'<\beta
    \\
    \Gamma:\OB
    \\
    A:\ITy{\alpha}{\Gamma}
  }{
    \Lift[\alpha'][\beta]{\Lift[\alpha][\alpha']{A}}
  }{
    \Lift[\alpha][\beta]{A}
  }{
    \ITy{\beta}{\Gamma}
  }
  \and
  \RuleSortAxiom[element lifting]{
    \alpha,\beta:\LVL
    \\
    p:\alpha<\beta
    \\
    \Gamma:\OB
    \\
    A:\ITy{\alpha}{\Gamma}
  }{
    \El{\Gamma}{A}
  }{
    \El{\Gamma}{\Lift[\alpha][\beta]{A}}
  }
  \and
  \RuleTermAxiom[context lifting]{
    \alpha,\beta:\LVL
    \\
    p:\alpha<\beta
    \\
    \Gamma:\OB
    \\
    A:\ITy{\alpha}{\Gamma}
  }{
    \Gamma.A
  }{
    \Gamma.\Lift[\alpha][\beta]{A}
  }{
    \OB
  }
\end{mathparpagebreakable}

\subsection{Type-theoretic connectives}

Adding connectives (like dependent function types) to the algebraic theory of
cumulative cwfs is simple, but we must take care to ensure that level shifting
commutes through the connectives properly. In the case of dependent function
types, we go beyond the usual only in adding enough axioms to make the shifts
``irrelevant'' (for instance, equating $\Lift[\alpha][\beta]{\TyPi{A}{B}}$ and
$\TyPi{\Lift[\alpha][\beta]{A}}{\Lift[\alpha][\beta]{B}}$).
\begin{mathparpagebreakable}
  \RuleOpDecl[pi formation]{\PI}{
    \alpha:\LVL
    \\
    \Gamma:\OB
    \\
    A:\ITy{\alpha}{\Gamma}
    \\
    B:\ITy{\alpha}{\Gamma.A}
  }{
    \ITy{\alpha}{\Gamma}
  }{\TyPi{A}{B}}
  \and
  \RuleTermAxiom[pi lifting]{
    \alpha,\beta:\LVL
    \\
    \Gamma:\OB
    \\
    A:\ITy{\alpha}{\Gamma}
    \\
    B:\ITy{\alpha}{\Gamma.A}
  }{
    \Lift[\alpha][\beta]{\TyPi{A}{B}}
  }{
    \TyPi{\Lift[\alpha][\beta]{A}}{\Lift[\alpha][\beta]{B}}
  }{
    \ITy{\beta}{\Gamma}
  }
  \and
  \RuleTermAxiom[pi substitution]{
    \alpha:\LVL
    \\
    \Delta,\Gamma:\OB
    \\
    \gamma:\Hom{\Delta}{\Gamma}
    \\
    A:\ITy{\alpha}{\Gamma}
    \\
    B:\ITy{\alpha}{\Gamma.A}
  }{
    \Act{\gamma}{\TyPi{A}{B}}
  }{
    \TyPi{\Act{\gamma}{A}}{
      \Act{\Snoc{\HomCmp{\gamma}{\PROJ}}{\VAR}}{B}
    }
  }{
    \ITy{\alpha}{\Delta}
  }
  \and
  \RuleOpDecl[pi introduction]{\LAM}{
    \alpha:\LVL
    \\
    \Gamma:\OB
    \\
    A:\ITy{\alpha}{\Gamma}
    \\
    B:\ITy{\alpha}{\Gamma.A}
    \\
    M:\El{\Gamma.A}{B}
  }{
    \El{\Gamma}{\TyPi{A}{B}}
  }{\Lam{M}}
  \and
  \RuleTermAxiom[lambda lifting naturality]{
    \alpha,\beta:\LVL
    \\
    p:\alpha<\beta
    \\
    \Gamma:\OB
    \\
    A:\ITy{\alpha}{\Gamma}
    \\
    B:\ITy{\alpha}{\Gamma.A}
    \\
    M:\El{\Gamma.A}{B}
  }{
    \Cut{\LAM}{\beta,\Gamma,\Lift[\alpha][\beta]{A},\Lift[\alpha][\beta]{B},M}
  }{
    \Cut{\LAM}{\alpha,\Gamma,A,B,M}
  }{
    \El{\Gamma}{\TyPi{A}{B}}
  }
  \and
  \RuleTermAxiom[lambda substitution]{
    \alpha:\LVL
    \\
    \Delta,\Gamma:\OB
    \\
    \gamma:\Hom{\Delta}{\Gamma}
    \\
    A:\ITy{\alpha}{\Gamma}
    \\
    B:\ITy{\alpha}{\Gamma.A}
    \\
    M:\El{\Gamma.A}{B}
  }{
    \Act{\gamma}{\Lam{M}}
  }{
    \Lam{\Act{\Snoc{\HomCmp{\gamma}{\PROJ}}{\VAR}}{M}}
  }{
    \El{\Delta}{
      \TyPi{\Act{\gamma}{A}}{
        \Act{\Snoc{\HomCmp{\gamma}{\PROJ}}{\VAR}}{B}
      }
    }
  }
  \and
  \RuleOpDecl[pi elimination]{\APP}{
    \alpha:\LVL
    \\
    \Gamma:\OB
    \\
    A:\ITy{\alpha}{\Gamma}
    \\
    B:\ITy{\alpha}{\Gamma.A}
    \\
    M:\El{\Gamma}{\TyPi{A}{B}}
    \\
    N:\El{\Gamma}{A}
  }{
    \El{\Gamma}{\Act{\Snoc{\HomId}{N}}{B}}
  }{\App{M}{N}}
  \and
  \RuleTermAxiom[app lifting naturality]{
    \alpha,\beta:\LVL
    \\
    p:\alpha<\beta
    \\
    \Gamma:\OB
    \\
    A:\ITy{\alpha}{\Gamma}
    \\
    B:\ITy{\alpha}{\Gamma.A}
    \\
    M:\El{\Gamma}{\TyPi{A}{B}}
    \\
    N:\El{\Gamma}{A}
  }{
    \Cut{\APP}{\beta,\Gamma,\Lift[\alpha][\beta]{A},\Lift[\alpha][\beta]{B},M,N}
  }{
    \Cut{\APP}{\alpha,\Gamma,A,B,M,N}
  }{
    \El{\Gamma}{\Act{\Snoc{\HomId}{N}}{B}}
  }
  \and
  \RuleTermAxiom[app substitution]{
    \alpha:\LVL
    \\
    \Delta,\Gamma:\OB
    \\
    \gamma:\Hom{\Delta}{\Gamma}
    \\
    A:\ITy{\alpha}{\Gamma}
    \\
    B:\ITy{\alpha}{\Gamma.A}
    \\
    M:\El{\Gamma}{\TyPi{A}{B}}
    \\
    N:\El{\Gamma}{A}
  }{
    \Act{\gamma}{\App{M}{N}}
  }{
    \App{\Act{\gamma}{M}}{\Act{\gamma}{N}}
  }{
    \El{\Delta}{
      \Act{\Snoc{\gamma}{\Act{\gamma}{N}}}{B}
    }
  }
  \and
  \RuleTermAxiom[pi unicity]{
    \alpha:\LVL
    \\
    \Gamma:\OB
    \\
    A:\ITy{\alpha}{\Gamma}
    \\
    B:\ITy{\alpha}{\Gamma.A}
    \\
    M:\El{\Gamma}{\TyPi{A}{B}}
  }{M}{
    \Lam{
      \App{
        \Act{\PROJ}{M}
      }{\VAR}
    }
  }{
    \El{\Gamma}{\TyPi{A}{B}}
  }
  \and
  \RuleTermAxiom[pi computation]{
    \alpha:\LVL
    \\
    \Gamma:\OB
    \\
    A:\ITy{\alpha}{\Gamma}
    \\
    B:\ITy{\alpha}{\Gamma.A}
    \\
    M:\El{\Gamma.A}{B}
    \\
    N:\El{\Gamma}{A}
  }{
    \App{M}{N}
  }{
    \Act{\Snoc{\HomId}{N}}{M}
  }{
    \El{\Gamma}{\Act{\Snoc{\HomId}{N}}{B}}
  }
\end{mathparpagebreakable}

\subsection{Universes \`a la Russell}
Now we see how to add universes \`a la Russell to the theory of cumulative
cwfs. First we add generators for the universe types themselves.

\begin{mathparpagebreakable}
  \RuleOpDecl[universe formation]{\UNIV}{
    \alpha,\beta:\LVL
    \\
    p:\alpha<\beta
    \\
    \Gamma:\OB
  }{
    \ITy{\beta}{\Gamma}
  }{\Univ{\alpha}}
  \and
  \RuleTermAxiom[universe lifting]{
    \alpha,\beta:\LVL
    \\
    p:\alpha<\alpha'
    \\
    q:\alpha<\beta
    \\
    \Gamma:\OB
  }{
    \Lift[\alpha'][\beta]{\Univ{\alpha}}
  }{
    \Univ{\alpha}
  }{
    \ITy{\beta}{\Gamma}
  }
  \and
  \RuleTermAxiom[universe substitution]{
    \alpha,\beta:\LVL
    \\
    p:\alpha<\beta
    \\
    \Delta,\Gamma:\OB
    \\
    \gamma:\Hom{\Delta}{\Gamma}
  }{
    \Act{\gamma}{\Univ{\alpha}}
  }{
    \Univ{\alpha}
  }{
    \ITy{\beta}{\Delta}
  }
\end{mathparpagebreakable}

To characterize the elements of the universes, it suffices to impose a sort
equation between the $\El{\Gamma}{\Univ{\alpha}}$ and $\ITy{\alpha}{\Gamma}$:
\begin{mathpar}
  \RuleSortAxiom[universe elements]{
    \alpha,\beta:\LVL
    \\
    p:\alpha<\beta
    \\
    \Gamma:\OB
  }{
    \El{\Gamma}{
      \Cut{\UNIV}{
        \alpha,\beta,p,\Gamma
      }
    }
  }{
    \ITy{\alpha}{\Gamma}
  }
\end{mathpar}

\subsection{A base type and constants}

We extend our theory with a base type and two constants; this will be useful
for illustrating a non-trivial metatheorem using algebraic methods.

\begin{mathparpagebreakable}
  \RuleOpDecl[formation]{\OBS}{
    \alpha:\LVL
    \\
    \Gamma:\OB
  }{
    \ITy{\alpha}{\Gamma}
  }{\OBS}
  \and
  \RuleTermAxiom[obs lifting]{
    \alpha,\beta:\LVL
    \\
    p:\alpha<\beta
    \\
    \Gamma:\OB
  }{
    \Lift[\alpha][\beta]{\OBS}
  }{
    \OBS
  }{
    \ITy{\beta}{\Gamma}
  }
  \and
  \RuleTermAxiom[obs substitution]{
    \alpha:\LVL
    \\
    \Delta,\Gamma:\OB
    \\
    \gamma:\Hom{\Delta}{\Gamma}
  }{
    \Act{\gamma}{\OBS}
  }{
    \OBS
  }{
    \ITy{\alpha}{\Delta}
  }
  \and
  \RuleOpDecl[introduction 1]{\RED}{
    \alpha:\LVL
    \\
    \Gamma:\OB
  }{
    \El{\Gamma}{\OBS}
  }{\RED}
  \and
  \RuleOpDecl[introduction 2]{\GREEN}{
    \alpha:\LVL
    \\
    \Gamma:\OB
  }{
    \El{\Gamma}{\OBS}
  }{\GREEN}
  \and
  \RuleTermAxiom[substitution 1]{
    \alpha:\LVL
    \\
    \Delta,\Gamma:\OB
    \\
    \gamma:\Hom{\Delta}{\Gamma}
  }{
    \Act{\gamma}{\RED}
  }{
    \RED
  }{
    \El{\Delta}{\OBS}
  }
  \and
  \RuleTermAxiom[substitution 2]{
    \alpha:\LVL
    \\
    \Delta,\Gamma:\OB
    \\
    \gamma:\Hom{\Delta}{\Gamma}
  }{
    \Act{\gamma}{\GREEN}
  }{
    \GREEN
  }{
    \El{\Delta}{\OBS}
  }
\end{mathparpagebreakable}

We do not add any elimination rules for the $\OBS$ type: its role is only to
serve as an \emph{observable} with which to phrase a canonicity result.

\section{Canonicity for Martin-L\"of Type Theory}

One of the major benefits of the algebraic approach to defining type theories
is that powerful \emph{semantic tools} are immediately available for developing
syntactic metatheory. As an example, we prove canonicity for Martin-L\"of Type
Theory with dependent function types and a hierarchy of universes \`a la
Russell, following \citet{coquand:2018}, \citet{shulman:2015} and
\citet{martin-lof:1975}. We will write $\SigMLTT$ for the theory defined in the
previous sections.

\subsection{The computability construction}

Starting from any model $\CC:\Alg{\SigMLTT}$, we will show how to construct a
\emph{new} model $\Gl{\CC}:\Alg{\SigMLTT}$ which \emph{glues} each context from $\CC$ together with a
\emph{logical family} over its closed elements, following
\citet{coquand:2018}.\footnote{We prefer the term ``logical family'' to the
more common term ``proof-relevant logical predicate''.} The logical families
construction is a modern and thoroughly constructive version of the method of
computability, in which predicates and relations are eschewed in favor of
proof-relevant families.

A consequence of this streamlined approach is that there is no need to consider
raw terms and partial equivalence relations, something which had previously
been essential for developing the syntactic metatheory of dependent type theory
with universes.

\begin{remark}[Categorical digression]

  While it is not our intention here to explain categorical gluing for
  dependent type theory (see \citet{coquand:2018} and \citet{shulman:2015}), we
  briefly observe that what we will unleash in type-theoretic language below
  can be understood more abstractly as an instance of the gluing construction,
  in which the fundamental fibration is pulled back along the global sections
  functor of the category of contexts:
  \begin{diagram}
    \DPullback{\Gl{\CC}}{
      \mathbf{Set}^{\to}
    }{\CC}{
      \mathbf{Set}
    }{
      \pi_{\mathsf{cod}}
    }{
      \CC(\cdot,-)
    }{}{}
  \end{diagram}

  The remainder of the work, then,  is to observe that the cwf structure lifts
  from $\CC$ to $\Gl{\CC}$. All aspects of the interpretation are forced
  \emph{except} for the interpretation of the base type and the universe, which
  we are free to choose; in fact, it is the choice of $\Gl{\CC}$-interpretation
  of the base type in which the essence of the proof of canonicity lies.

\end{remark}

\begin{assumption}[Set-theoretic universes]
  As a simplifying move, we assume a transfinite hierarchy of
  Grothendieck universes $\SemUniv{\alpha}$ in the ambient set theory for
  $\alpha\in\braces*{0,1,\ldots,\omega}$. The universe $\SemUniv{\omega}$ is only a
  convenience, and could be eliminated using a more verbose schematic or
  fibered construction.
\end{assumption}

To exhibit $\Gl{\CC}$ as a model, we must determine families of sets
$\SemSort[\Gl{\CC}]{\vartheta}$ to interpret each sort symbol
$\vartheta:\SortDecl{\Psi}\in\SigMLTT$; and likewise for operation symbols.
We begin by imposing some notation:
\begin{enumerate}

  \item We will write $\Dom{\Gamma}$ for
    $\SemSort[\CC]{\HOM}{\SemOp[\CC]{\cdot},\Gamma}$ when
    $\Gamma\in\SemSort[\CC]{\OB}$.

  \item We will write $\Dom{A}$ for $\SemSort[\CC]{\EL}{\alpha,\SemOp[\CC]{\cdot},A}$ when
    $A\in\SemSort[\CC]{\TY}{\alpha,\SemOp[\CC]{\cdot}}$.

  \item When it is unambiguous, we will use the notation of the algebraic
    theory itself to refer to objects in the models $\CC$ and $\Gl{\CC}$; for
    instance, $\Hom{\Delta}{\Gamma}$ rather than
    $\SemOp[\CC]{\HOM}{\Delta,\Gamma}$, and $\Hom{\Gl{\Delta}}{\Gl{\Gamma}}$
    rather than $\SemOp[\Gl{\CC}]{\HOM}{\Gl{\Delta},\Gl{\Gamma}}$. To avoid
    ambiguity, we will write the $\Gl{\CC}$-interpretations of unary sorts and
    operations with an overline, e.g.\ $\Gl{\OB}$ for
    $\SemSort[\Gl{\CC}]{\OB}$.

\end{enumerate}

As a notational convenience, we will generally write $\Gl{X}\equiv(X,\LFam{X})$
for something in the computability model, the gluing of $X$ from $\CC$ with its
``realizer'' $\LFam{X}$.

\paragraph{Contexts and substitutions}
We are now ready to define the glued interpretation of the contexts and substitutions.
\begin{align*}
  \Gl{\OB} &=
  \textstyle
  \coprod_{
    \Gamma\in\OB
  }{
    \parens*{\Dom{\Gamma}\to\SemUniv{\omega}}
  }
  \tag{objects}
\end{align*}

We will write $\Gl{\Gamma}$ for
$\parens*{\Gamma,\LFam{\Gamma}}\in\Gl{\OB}$. Substitutions are
interpreted as $\CC$-substitutions together with realizers of the logical
family.
\begin{align*}
  \Hom{\Gl{\Delta}}{\Gl{\Gamma}}
  &=
  \textstyle
  \coprod_{
    \gamma\in\Hom{\Delta}{\Gamma}
  }
  \prod_{
    \delta\in\Dom{\Delta}
  }
  \prod_{
    \LFam{\delta}\in\LFam{\Delta}\delta
  }
  \LFam{\Gamma}{
    \parens*{\HomCmp{\gamma}{\delta}}
  }
  \tag{substitutions}
\end{align*}

When interpreting operations $\Gl{\vartheta}$ whose sort-interpretations in
$\Gl{\CC}$ glue a $\CC$-construct together with a realizer, we follow
convention of merely exhibiting the realizer $\LFam{\vartheta}$ rather than
$\parens*{\vartheta, \LFam{\vartheta}}$.
\begin{align*}
  \LFam{\HOMID}\gamma\LFam{\gamma} &= \LFam{\gamma}
  \tag{identity substitution}
  \\
  \LFam{\parens*{\HomCmp{\Gl{\gamma}}{\Gl{\delta}}}}\eta\LFam{\eta} &=
  \LFam{\gamma}\parens*{\HomCmp{\delta}{\eta}}\parens*{\LFam{\delta}\eta\LFam{\eta}}
  \tag{substitution composition}
\end{align*}

\paragraph{Levels}

First, observe that we can embed any natural number $\alpha$ as a term of sort
$\LVL$ in $\SigMLTT$; write $\lfloor\alpha\rfloor$ for its interpretation in
$\CC$. Then, we interpret $\LVL$ in the computability model as a $\CC$-level
together with a compatible natural number (its ``realizer''):
\begin{align*}
  \Gl{\LVL} &=
  \textstyle
  \coprod_{
    \alpha\in\LVL
  }
  \braces*{\LFam{\alpha} \in\Nat \mid \lfloor \LFam{\alpha} \rfloor = \alpha}
  \\
  \Gl{\LZ} &= \parens*{\LZ,0}
  \\
  \Cut{\LS}{\Gl{\alpha}} &= \parens*{\Cut{\LS}{\alpha}, \LFam{\alpha}+1}
\end{align*}

Writing $\Gl{\alpha}$ for $\parens*{\alpha,\LFam{\alpha}}\in\Gl{\LVL}$, we
interpret the order on levels in the obvious way:
\begin{align*}
  \parens*{\Gl{\alpha} < \Gl{\beta}}
  &=
  \braces*{p: \alpha < \beta \mid \LFam{\alpha}<\LFam{\beta}}
\end{align*}

The interpretation is always subsingleton, so it validates the proof
irrelevance axiom that we imposed on $\LVL$.

\paragraph{Types and elements}

A type is interpreted as a $\CC$-type together with a logical family of the
appropriate size defined over its closed instances:
\begin{align*}
  \ITy{\Gl{\alpha}}{\Gl{\Gamma}}
  &=
  \textstyle
  \coprod_{
    A\in\ITy{\alpha}{\Gamma}
  }
  \prod_{
    \gamma\in\Dom{\Gamma}
  }
  \prod_{
    \LFam{\gamma}\in\LFam{\Gamma}\gamma
  }
  \parens*{\Dom{\Act{\gamma}{A}} \to \SemUniv{\LFam{\alpha}}}
  \tag{types}
  \\
  \El{\Gl{\Gamma}}{\Gl{A}}
  &=
  \textstyle
  \coprod_{M\in\El{\Gamma}{A}}
  \prod_{
    \gamma\in\Dom{\Gamma}
  }
  \prod_{
    \LFam{\gamma}\in\LFam{\Gamma}\gamma
  }
  \LFam{A}\gamma\LFam{\gamma}\Act{\gamma}{M}
  \tag{elements}
\end{align*}

\paragraph{Terminal context}\hfill
\begin{align*}
  \LFam{\cdot}\eta &= \braces*{\star}
  \tag{terminal context}
  \\
  \LFam{\BANG}\eta\LFam{\eta} &= \star
  \tag{universal substitution}
\end{align*}

\paragraph{Context comprehension}\hfill
\begin{align*}
  \LFam{\Gl{\Gamma}.\Gl{A}}\eta
  &=
  \textstyle
  \coprod_{
    \LFam{\gamma}\in \LFam{\Gamma}\parens*{\HomCmp{\PROJ}{\eta}}
  }
  \LFam{A}\parens*{\HomCmp{\PROJ}{\eta}}\LFam{\gamma}\parens*{\Act{\eta}{\VAR}}
  \tag{context extension}
  \\
  \LFam{
    \Snoc{\Gl{\gamma}}{\Gl{N}}
  }
  \delta
  \LFam{\delta}
  &=
  \parens*{\LFam{\gamma}\delta\LFam{\delta}, \LFam{N}\delta\LFam{\delta}}
  \tag{substitution extension}
  \\
  \Gl{\PROJ}\eta\parens*{\LFam{\gamma},\LFam{N}} &= \LFam{\gamma}
  \tag{projection}
  \\
  \Gl{\VAR}\eta\parens*{\LFam{\gamma},\LFam{N}} &= \LFam{N}
  \tag{variables}
\end{align*}

\paragraph{Type lifting}
The following interpretation of type lifting is well-defined, because we have
imposed the axiom that the lifting of a type has the same elements as the type
itself, and because the set-theoretic universe hierarchy $\SemUniv{\alpha}$ is
cumulative.
\begin{align*}
  \LFam{\parens*{\Lift[\Gl{\alpha}][\Gl{\beta}]{\Gl{A}}}}
  \gamma\LFam{\gamma}M &=
  \LFam{A}\gamma\LFam{\gamma}M
  \tag{type lifting}
\end{align*}

\paragraph{Dependent function types}\hfil
\begin{align*}
  \LFam{\TyPi{\Gl{A}}{\Gl{B}}}\gamma\LFam{\gamma}M &=
  \textstyle
  \prod_{
    N\in\Dom{\Act{\gamma}{A}}
  }
  \prod_{
    \LFam{N}\in\LFam{A}\gamma\LFam{\gamma}N
  }
  \LFam{B}\Snoc{\gamma}{N}\parens*{\LFam{\gamma},\LFam{N}}
  \parens*{\App{M}{N}}
  \tag{formation}
  \\
  \LFam{\Lam{\Gl{M}}}\gamma\LFam{\gamma}N\LFam{N} &=
  \LFam{M}\Snoc{\gamma}{N}\parens*{\LFam{\gamma},\LFam{N}}
  \tag{introduction}
  \\
  \LFam{\App{\Gl{M}}{\Gl{N}}}\gamma\LFam{\gamma}
  &=
  \LFam{M}\gamma\LFam{\gamma}N\LFam{N}
  \tag{elimination}
\end{align*}

\paragraph{Universes \`a la Russell}

A realizer for an element of the $\Gl{\alpha}$th universe is an
$\LFam{\alpha}$-small logical family over its closed elements:
\begin{align*}
  \LFam{\Univ{\Gl{\alpha}}}\gamma\LFam{\gamma}A &=
  \textstyle
  \Dom{A}\to\SemUniv{\LFam{\alpha}}
  \tag{formation}
\end{align*}

Fixing $\Gl{\alpha,\beta}\in\Gl{\LVL}$ with $\Gl{p} : \Gl{\alpha}<\Gl{\beta}$
and $\Gl{\Gamma}\in\Gl{\OB}$, we need to see that
$\El{\Gl{\Gamma}}{\Univ{\Gl{\alpha}}} = \ITy{\Gl{\alpha}}{\Gl{\Gamma}}$. Calculate:
\begin{align*}
  &\El{\Gl{\Gamma}}{\Univ{\Gl{\alpha}}}
  \\
  &\quad=
  \textstyle
  \coprod_{
    A\in\El{\Gamma}{\Univ{\alpha}}
  }
  \prod_{
    \gamma\in\Dom{\Gamma}
  }
  \prod_{
    \LFam{\gamma}\in\LFam{\Gamma}\gamma
  }
  \LFam{\Univ{\Gl{\alpha}}}\gamma\LFam\gamma \Act{\gamma}{A}
  \tag{by def.}
  \\
  &\quad=
  \textstyle
  \coprod_{
    A\in\ITy{\alpha}{\Gamma}
  }
  \prod_{
    \gamma\in\Dom{\Gamma}
  }
  \prod_{
    \LFam{\gamma}\in\LFam{\Gamma}\gamma
  }
  \LFam{\Univ{\Gl{\alpha}}}\gamma\LFam\gamma \Act{\gamma}{A}
  \tag{$\CC$ has universes}
  \\
  &\quad=
  \textstyle
  \coprod_{
    A\in\ITy{\alpha}{\Gamma}
  }
  \prod_{
    \gamma\in\Dom{\Gamma}
  }
  \prod_{
    \LFam{\gamma}\in\LFam{\Gamma}\gamma
  }
  \parens*{\Dom{\Act{\gamma}{A}}\to\SemUniv{\LFam{\alpha}}}
  \tag{by def.}
  \\
  &\quad=
  \ITy{\Gl{\alpha}}{\Gl{\Gamma}}
  \tag{by def.}
\end{align*}

\paragraph{Base type}

Finally, we give the interpretation of the base type.
\begin{align*}
  \LFam{\OBS}\gamma\LFam{\gamma}M &=
  \braces*{\star\mid M = \RED}
  +
  \braces*{\star\mid M = \GREEN}
  \\
  \LFam{\RED} \gamma\LFam{\gamma} &=
  \Inl{\star}
  \\
  \LFam{\GREEN} \gamma\LFam{\gamma} &=
  \Inr{\star}
\end{align*}

\paragraph{Projection from the computability model}

There is an evident projection $\pi : \Gl{\CC}\rightarrowtriangle\CC$ which
merely forgets the realizers; it is easy to see that it is a homomorphism of
$\SigMLTT$-algebras.

\subsection{The canonicity theorem}
In this section, we let $\CC$ be the \emph{initial} model for $\SigMLTT$.

\begin{theorem}[Canonicity]

  If $\IsTerm[\SigMLTT]{\cdot}{M}{\El{\cdot}{\OBS}}$, then
  either $M = \RED$ or $M=\GREEN$.

\end{theorem}

\begin{proof}

  Because $\Gl{\CC}$ is a model of $\SigMLTT$, we have some $\parens*{N,\LFam{N}}$ in
  the interpretation of $M$ such that $\pi M = N$ and $\LFam{N}(\cdot)(\star)$ is
  an element of the following set:
  \[
    \braces*{\star \mid N = \RED}
    +
    \braces*{\star \mid N = \GREEN}
  \]
  Therefore, it suffices to observe that $M=N$; this is guaranteed by the
  universal property of the initial $\SigMLTT$-algebra $\CC$ and the fact that
  $\pi$ is a homomorphism of $\SigMLTT$-algebras.

\end{proof}

\section*{Acknowledgements}

Thanks to Carlo Angiuli, John Cartmell, David Thrane Christiansen, Thierry
Coquand, Daniel Gratzer, Robert Harper, Ambrus Kaposi, Darin Morrison, Anders M\"ortberg,
Michael Shulman, and Thomas Streicher for helpful conversations about algebraic
type theory, logical frameworks and the method of computability. We also thank
Paul Taylor for his \verb|diagrams| package, which we have used to typeset the
commutative diagrams in this note. Special thanks to David Thrane Christiansen,
Daniel Gratzer and Jacques Carette for comments on an earlier version of this
note.

\nocite{streicher:1991}
\nocite{hofmann:1997}
\nocite{abel-coquand-dybjer:2008}
\nocite{buisse-dybjer:2008}
\nocite{palmgren-vickers:2007}
\nocite{martin-lof:1984}
\nocite{nordstrom-peterson-smith:1990}
\nocite{harper-honsell-plotkin:1993}
\nocite{fiore:2002}
\nocite{shulman:blog:scones-logical-relations}
\nocite{streicher:1998}
\nocite{martin-lof:itt:1975}
\nocite{martin-lof:1986}
\nocite{harper:2012}
\nocite{martin-lof:1994}
\nocite{martin-lof:1987:wgl}
\nocite{jacobs:1999}

\bibliographystyle{plainnat}
\bibliography{references/refs-bibtex}
\end{document}